\documentclass[pra,aps,twocolumn,superscriptaddress]{revtex4-1}

\usepackage{amsmath,amsfonts,amssymb,amsthm,mathtools}
\usepackage[locale=US]{siunitx}
\usepackage{graphicx}
\usepackage{color}
\usepackage{hyperref}

\newcommand{\ssmile}[1]{%
	\mathop{%
		\sbox0{$\displaystyle\smile$}%
		\raisebox{-0.5em}{\scalebox{#1}[1]{\copy0}}}	
	\displaylimits}
\newtheorem{lemma}{Lemma}

\begin{document}

\title{Interferometric sensitivity and entanglement by scanning through quantum phase transitions in spinor Bose-Einstein condensates}

\author{P. Feldmann}
\affiliation{Institut f\"ur Theoretische Physik, Leibniz Universit\"at Hannover, Appelstr. 2, DE-30167 Hannover, Germany}
\author{M. Gessner}
\affiliation{QSTAR, INO-CNR, and LENS, Largo Enrico Fermi 2, IT-50125 Firenze, Italy}
\author{M. Gabbrielli}
\affiliation{QSTAR, INO-CNR, and LENS, Largo Enrico Fermi 2, IT-50125 Firenze, Italy}
\affiliation{Dipartimento di Fisica e Astronomia, Universit\`a degli Studi di Firenze, via Sansone 1, I-50019 Sesto Fiorentino, Italy}
\author{C. Klempt}
\affiliation{Institut f\"ur Quantenoptik, Leibniz Universit\"at Hannover, Welfengarten 1, DE-30167 Hannover, Germany}
\author{L. Santos}
\affiliation{Institut f\"ur Theoretische Physik, Leibniz Universit\"at Hannover, Appelstr. 2, DE-30167 Hannover, Germany}
\author{L. Pezz\`{e}}
\affiliation{QSTAR, INO-CNR, and LENS, Largo Enrico Fermi 2, IT-50125 Firenze, Italy}
\author{A. Smerzi}
\affiliation{QSTAR, INO-CNR, and LENS, Largo Enrico Fermi 2, IT-50125 Firenze, Italy}

\begin{abstract}
Recent experiments have demonstrated the generation of entanglement by quasi-adiabatically driving through quantum phase transitions of a ferromagnetic spin-$1$ Bose-Einstein condensate in the presence of a tunable quadratic Zeeman shift. We analyze, in terms of the Fisher information, the interferometric value of the entanglement accessible by this approach. In addition to the Twin-Fock phase studied experimentally, we unveil a second regime, in the broken axisymmetry phase, which provides Heisenberg scaling of the quantum Fisher information and can be reached on shorter time scales. We identify optimal unitary transformations and an experimentally feasible optimal measurement prescription that maximize the interferometric sensitivity. We further ascertain that the Fisher information is robust with respect to non-adiabaticity and measurement noise. Finally, we show that the quasi-adiabatic entanglement preparation schemes admit higher sensitivities than dynamical methods based on fast quenches.
\end{abstract}


\maketitle


\section{Introduction}
\label{sec:Intro}
Atom interferometry has become an indispensable tool for both the testing of 
fundamental physics and precision measurements~\cite{Pritchard2009}. 
Without entanglement between the atoms, the attainable sensitivity is fundamentally limited by 
the standard quantum limit (SQL)~\cite{GiovannettiPRL2006, PezzePRL2009}. 
Employing multipartite entanglement allows to shift this bound 
towards the Heisenberg limit (HL)~\cite{GiovannettiPRL2006, PezzePRL2009, HyllusPRA2012}. 
In view of the high effort required for handling coherent ensembles with a large atom number $N$, it is crucial that the HL replaces the SQL scaling of the sensitivity $\propto \sqrt{N}$  with $\propto N$. Entanglement 
that---in the absence of technical noise---facilitates to surpass the SQL 
is unambiguously witnessed by the Fisher information (FI).

In spinor Bose-Einstein condensates (BEC), entanglement useful to enhance the
sensitivity of atom interferometers beyond the SQL can 
be generated exploiting spin-changing collisions~\cite{PezzeRMP}.
A common realization relies on the parametric amplification of quantum fluctuations 
leading to squeezed Gaussian states~\cite{LuckeSCIENCE2011, GrossNATURE2011, 
HamleyNATPHYS2012, LuckePRL2014, PeiseNATCOMM2015, KrusePRL2016, Oberthaler2017}.
Sub-SQL sensitivities~\cite{LuckeSCIENCE2011, KrusePRL2016, Oberthaler2017}, entanglement~\cite{LuckePRL2014, PeiseNATCOMM2015}, and squeezing of up to 
$\SI{10}{dB}$ beyond the SQL~\cite{PeiseNATCOMM2015, GrossNATURE2011, HamleyNATPHYS2012} 
have been demonstrated~\cite{PezzeRMP}.  
Furthermore, multipartite entanglement in spinor BECs can be also generated near the ground state of 
ferromagnetic~\cite{ZhangPRL2013} and antiferromegnetic~\cite{WuPRA2016} spin-$1$ BECs~\cite{KajtochPREPRINT}.
In the following we focus on the ferromagnetic case, relevant for experiments with $^{87}$Rb.
In the presence of an (effective) quadratic Zeeman shift $q$, the system
exhibits three quantum phases \cite{Stamper-KurnRMP2013,ZhangPRL2013}.
By preparing the BEC in the ground state at $q>q_c$ 
(here $q_c>0$ and the critical points are located at $q=\pm q_c$)
and slowly driving through both quantum phase transitions 
an entanglement depth of about \num{900} particles 
has been witnessed in the Twin-Fock (TF) phase at $q < -q_c$~\cite{LuoSCIENCE2017}, 
see also~\cite{HoangPNAS2016, VinitPRA2017}.

Recently \cite{shortversion}, we have shown that the ground state of a ferromagnetic at $q=0$ 
can be used for heralded generation of highly entangled macroscopic superposition states. 
In the present paper we extend this study and analyze the interferometric sensitivity of the entangled quantum states 
that are generated along the \mbox{(quasi-)}adiabatic passage when scanning over different values of $q$. 
Their full potential is revealed by considering atom interferometry involving all three modes, which generalizes the two-mode interferometry experimentally implemented, e.\,g., in \cite{KrusePRL2016}. We find that besides the TF 
state studied in Ref.~\cite{LuoSCIENCE2017, ZhangPRL2013} also 
the ground state at the center, $q=0$, of the broken axisymmetry phase leads to Heisenberg scaling. This state can be reached by \mbox{(quasi-)}adiabatically scanning over only a single critical point, stopping the evolution half-way to the TF state. We identify the interferometric transformations that provide the most sensitive phase imprinting and demonstrate that the measurement of particle numbers, 
an established experimental technique, is optimal for the phase estimation.
Our simulations show that performing the passage within reasonable, finite time does not strongly impair the attainable FI. We further analyze the effect of measurement noise and find that surpassing the SQL with state-of-the-art technology is well feasible. We finally show that, under realistic conditions, quasi-adiabatic schemes produce states with larger interferometric sensitivity than those accessible by parametric amplification.


\section{Fisher Information and Interferometry}
\label{sec:FI}
Let us briefly review some concepts used in the paper. 
In any atom interferometer, a phase $\theta$ is imprinted into an initial 
density matrix $\hat{\rho}_0$, leading to a $\hat{\rho}_\theta$ which is subsequently measured to determine the phase $\theta$. 
The resulting estimation of $\theta$ has an uncertainty, which is bounded by the (classical) Cram\'{e}r-Rao bound, $\Delta\theta\geq \Delta\theta_{\mathrm{CR}}$.  
Here 
\begin{equation}
\Delta \theta_{\mathrm{CR}} =1/\sqrt{\nu F} 
\end{equation}
and
\begin{equation}
F(\theta)=\sum_{\mu}\frac{1}{P(\mu|\theta)}\left(\frac{\partial P(\mu | \theta)}{\partial \theta}\right)^2
\label{eq:CFI}
\end{equation}
is the (classical) FI which depends on $\hat{\rho}_{\theta}$ and the chosen measurement observable. The sum comprises all possible measurement outcomes $\mu$ and $P(\mu | \theta)$ is the probability to measure $\mu$ given that the quantum state is $\hat{\rho}_{\theta}$. Finally,
$\nu$ is the number of measurements~\cite{Smerzi2014}. 
Maximizing the FI over all possible generalized quantum measurements
defines the quantum Fisher information (QFI) $F_Q$~\cite{Caves1994, Helstrom1976, Smerzi2014}:
$F\leq F_Q$ and the equality $F =  F_Q$ can always be reached by 
an optimal measurement~\cite{Caves1994}. 
Correspondingly, a quantum Cram\'{e}r-Rao bound is introduced as 
\begin{equation}
\Delta \theta_{\mathrm{QCR}} = 1/\sqrt{\nu F_Q}
\end{equation}
with $\Delta\theta_{\mathrm{CR}}\geq \Delta\theta_{\mathrm{QCR}}$.
For $N$ qubits, we have $F_Q\leq N^2$~(HL) \cite{GiovannettiPRL2006, PezzePRL2009}, and 
$F_Q \leq N$~(SQL) if $\hat{\rho}_0$ is not entangled~\cite{PezzePRL2009}. 
Thus both the classical and quantum FI witness interferometrically useful entanglement: 
$F_Q \geq F > N$ is equivalent to a $\Delta \theta_{\mathrm{QCR}}$ undercutting the SQL.

We assume that the phase $\theta$ is imprinted by a collective unitary transformation across $n$ modes. 
Let $\hat{g}_j$ be the generators of the defining representation of $\mathfrak{su}(n)$. We denote the vector comprising these $d_n=n^2-1$ generators with respect to the $i$-th of $N$ particles as $\mathbf{\hat{g}}^{(i)}\equiv (\hat g_1^{(i)},\ldots,\hat g_{d_n}^{(i)} )$. Then the final density matrix acquires the form
$\hat{\rho}_\theta = \hat{U}(\theta)\hat{\rho}_0 \hat{U}^\dagger\!(\theta)$, with
$\hat{U}(\theta) = \exp (-i\theta\, \mathbf{u}\cdot\mathbf{\hat{G}} )$, where 
we call $\mathbf{\hat{G}} = \sum_{i=1}^{N}\mathbf{\hat{g}}^{(i)}$ the collective $\mathbf{\hat g}$, and $\mathbf{u}\in S^{d_n-1}$  
is the interferometric direction.
For a pure initial state, $\hat{\rho}_0=|\psi\rangle\langle\psi|$, the QFI due to an interferometric transformation generated by $\hat{R}_{\mathbf{u}}\equiv\mathbf{u}\cdot\mathbf{\hat{G}}$ reads
\begin{equation}
F_Q[|\psi\rangle,\,\hat{R}_{\mathbf{u}}]=4\,(\Delta \hat{R})^2=4\, \mathbf{u}^T\!\boldsymbol{\Gamma}_{\mathbf{\hat{G}}}\;\!\mathbf{u},
\end{equation}
where $\boldsymbol{\Gamma}_{\mathbf{\hat{G}}}$ denotes the covariance matrix of the operators composing $\mathbf{\hat G}$ \cite{Gessner2016}, with elements $(\boldsymbol{\Gamma}_{\hat{\mathbf{G}}})_{ij}=\langle \hat G_i \hat G_j \rangle/2 + \langle \hat G_j \hat G_i \rangle/2 - \langle \hat G_i\rangle \langle \hat G_j\rangle$. The leading eigenvector of $\boldsymbol{\Gamma}_{\mathbf{\hat{G}}}$
identifies the optimal interferometric direction $\mathbf{u}_\mathrm{opt}$. By convention, in the case of qubits ($n=2$) the $\hat g_i$ are normalized such that $(\gamma_{i\max}-\gamma_{i\min})^2=1$, $\gamma_{i\max}$ and $\gamma_{i\min}$ being the maximum and minimum eigenvalues of $\hat g_i$, respectively. More generally, the SQL is given by $(\gamma_{i\max}-\gamma_{i\min})^2N$ and the HL by $(\gamma_{i\max}-\gamma_{i\min})^2N^2$~\cite{GiovannettiPRL2006}.


\section{Model}
\label{sec:Model}

In the following we study an optically trapped spin-$1$ BEC of 
$N$ particles with magnetic sublevels $m_f\in\{0,\pm 1\}$. 
In the single-mode approximation, the spinor dynamics is modeled by the Hamiltonian~\cite{LawPRL1998, UedaPR2012}
\begin{align}
\begin{split}
\hat{H} =& \Big[ \lambda \Big(\hat{N}_0-\frac{1}{2} \Big) +q \Big] (\hat{N}_++\hat{N}_-) 
\\
&+ \lambda (\hat{a}_1^\dagger \hat{a}_{-1}^\dagger \hat{a}_0^2 + \hat a_0^{\dagger 2}\hat a_1 \hat a_{-1}),
\label{eq:Hamiltonian}
\end{split}
\end{align} 
where $\hat{a}_i^\dagger$ and $\hat{a}_i$ are the creation and annihilation operators for $m_f=i$, and
$\hat{N}_{0,\pm} = \hat{a}^\dagger_{0,\pm 1} \hat{a}_{0,\pm1}$
are the number operators for the respective sublevels.
The total number of atoms is equal to $N$ and is assumed fixed here.
The interaction coefficient $\lambda$~(negative for ferromagnetic condensates 
such as the $F=1$ hyperfine groundstate manifold of $^{87}$Rb) depends on the trapping potential and microscopic 
parameters, namely the scattering lengths and the mass of the atoms~\cite{UedaPR2012, nota}. 
The effective quadratic Zeeman shift $q$ may be controlled by an external magnetic field 
and near-resonant microwave dressing~\cite{Stamper-KurnRMP2013, UedaPR2012}. 
Spin-changing collisions, described by the last line of Eq.~\eqref{eq:Hamiltonian}, 
preserve the total magnetization, i.\,e., the eigenvalue $D$ of $\hat{D}\equiv \hat{N}_+-\hat{N}_-$. 
Hence starting from an initial condensate in $m_f=0$ and then quenching---or slowly driving---the magnetic field so to 
prepare entangled states ensures that the system remains in the subspace of $D=0$. 
The dynamics thus takes place in the Hilbert space spanned by the Fock states $|k\rangle\equiv |N_-=k,\;N_0=N-2k,\;N_+=k\rangle$ with $N_i$ the eigenvalues of $\hat{N}_i$. By restricting the dynamics to the magnetization-free subspace, the linear coupling to the magnetic field and its fluctuations (linear Zeeman shift) becomes irrelevant, which leads to phase noise stability.

In the magnetization-free subspace, model~\eqref{eq:Hamiltonian} presents three quantum phases~\cite{ZhangPRL2013, Stamper-KurnRMP2013} as a function of $q$ with quantum phase transitions at $q=\pm q_c$, $q_c = 2N |\lambda|$:  
the polar~(P) phase~($q>q_c$), the broken-axisymmetry~(BA) phase~($|q|<q_c$), and the TF phase~($q<-q_c$). For large $N$, the respective ground states approach $|k=0\rangle$ in the P phase and the TF state $|\!\operatorname{TF}\rangle\equiv|k=N/2\rangle$ in the TF phase. In the BA phase,
all the three modes stay populated, with an average number of particles in $m_f=0$
given by $\langle \hat{N}_0 \rangle/N \simeq (1+q/q_c)/2$~\cite{PezzeRMP}. 


\section{Useful entanglement in the ground state of a spinor BEC}
\label{sec:PT}

\begin{figure}[t]
	\begin{center}
		\includegraphics[width=\columnwidth]{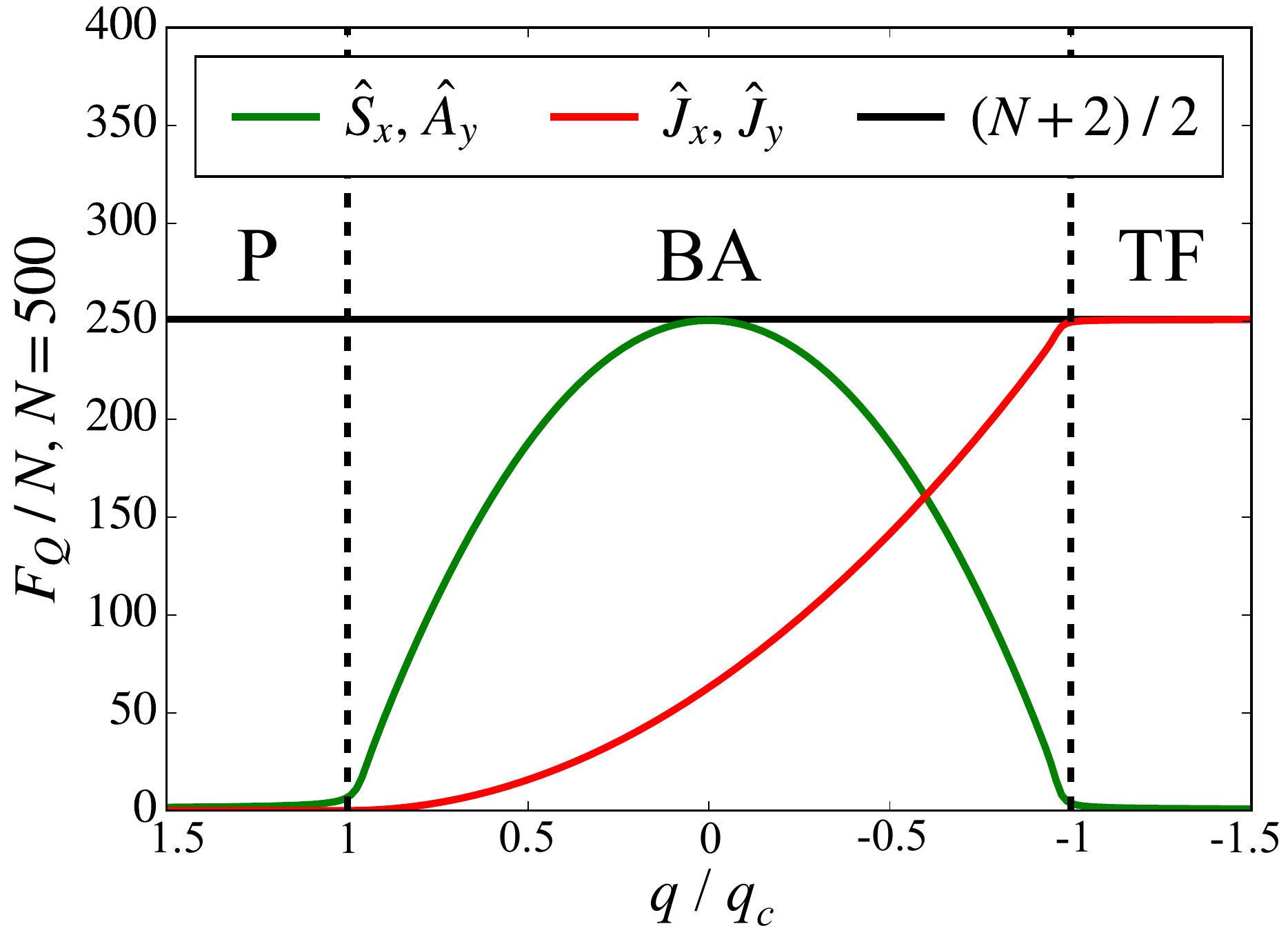}
	\end{center}
	\caption{(Color online) QFI for the ground state of a non-magnetized ferromagnetic spin-$1$ 
	BEC as a function of the quadratic Zeeman shift $q$. 
	For $\mathbf{\hat{G}}$ generating the $\mathfrak{su}(3)$, we depict 
	$F_Q/N$ corresponding to the eigendirections of $\boldsymbol{\Gamma}_{\mathbf{\hat{G}}}$ 
	for which $F_Q$ exceeds $N$. The quantum phase transitions are indicated by dashed vertical lines. 
	The solid horizontal line is $F_Q/N=(N+2)/2$.}
	\label{fig:1}
\end{figure}

We first evaluate the QFI of the ground state $|\psi_0(q)\rangle$ of the Hamiltonian~(\ref{eq:Hamiltonian}) in the different phases. Arbitrary collective unitary rotations of a system of indistinguishable spin-1 particles, as considered in this paper, can be expressed by taking 
$\mathbf{\hat{G}}$ as the 8-dimensional vector of collective Gell-Mann operators, generating the $\mathfrak{su}(3)$. The covariance matrix $\mathbf{\Gamma_{\hat G}}$ is discussed in Appendix~\ref{app:GS_Cov}. 
We find it convenient to introduce the symmetric ($g$) and antisymmetric ($h$) creation and annihilation operators
\begin{equation}
\hat{g}^\dagger=\frac{1}{\sqrt{2}}(\hat{a}_1^\dagger+\hat{a}_{-1}^\dagger),\;\;\hat{h}^\dagger=\frac{1}{\sqrt{2}}(\hat{a}_1^\dagger-\hat{a}_{-1}^\dagger)
\end{equation}
and present our results in terms of three sets of collective pseudospin-$\frac{1}{2}$ operators, whose Schwinger representation reads
\begin{equation}
\begin{alignedat}{2} \label{SAoperators}
\hat{S}_x &= \frac{\hat{a}_0^\dag\hat{g} + \hat{g}^\dag\hat{a}_0}{2},
\hspace{2.5em}\hat{J}_x &= \frac{\hat{a}_1^\dag\hat{a}_{-1} + \hat{a}_{-1}^\dag\hat{a}_1 }{2},\\
\hat{S}_y &= \frac{\hat{a}_0^\dag\hat{g} - \hat{g}^\dag\hat{a}_0}{2i},
\hspace{2.5em}\hat{J}_y &= \frac{\hat{a}_1^\dag\hat{a}_{-1} - \hat{a}_{-1}^\dag\hat{a}_1 }{2i},\\
\hat{S}_z &= \frac{\hat{a}_0^\dag\hat{a}_0 -\hat{g}^\dag\hat{g}}{2},
\hspace{2.5em}\hat{J}_z &= \frac{\hat{a}_1^\dag\hat{a}_1 -\hat{a}_{-1}^\dag\hat{a}_{-1}}{2},
\end{alignedat}
\end{equation}
and $\hat{A}_i$ just as $\hat{S}_i$ with $\hat{g}$ replaced by $\hat{h}$. Thus $\mathbf{\hat S}\equiv (\hat{S}_1,\hat S_2,\hat S_3)$ generates rotations within the two-level system composed of the modes $(\hat{a}_0,\hat{g})$, $\mathbf{\hat A}\equiv (\hat A_1,\hat A_2,\hat A_3)$ corresponds to $(\hat{a}_0,\hat{h})$, and $\mathbf{\hat J}\equiv (\hat J_x,\hat J_y, \hat J_z)$ to $(\hat{a}_1,\hat{a}_{-1})$. Figure~\ref{fig:1} displays, across the three quantum phases P, BA, and TF, $F_Q/N$ corresponding to the eigenvectors of $\mathbf{\Gamma_{\hat G}}$ which provide $F_Q>N$. Large values of the QFI are observed in two cases. 
First, in the TF phase,
\begin{equation} \label{QFITF}
F_Q\Big[|\!\operatorname{TF}\rangle,\,\hat{R}_\mathrm{opt}^{(\mathrm{TF})}\Big]=N(N+2)/2 
\end{equation}
where $\hat{R}_\mathrm{opt}^{(\mathrm{TF})}$ is given by an arbitrary linear combination of $\hat{J}_x$  and $\hat{J}_y$.
Second, at the center of the BA phase (i.\,e., for $q=0$, we indicate as $|\!\operatorname{CBA}\rangle$
the corresponding ground state) we have 
\begin{equation} \label{QFICBA}
F_Q\left[|\!\operatorname{CBA}\rangle,\, \hat{R}_\mathrm{opt}^{(\mathrm{CBA})}\right] = N(N+1)/2,
\end{equation}
where $\hat{R}_\mathrm{opt}^{(\mathrm{CBA})}$ is an arbitrary linear combination of 
$\hat{S}_x$ and $\hat{A}_y$.
As we show in Appendix~\ref{app:GS_q0}, the state has an explicit expression given by
\begin{equation} 
|\!\operatorname{CBA}\rangle\equiv \sqrt{\frac{2^N (N!)^3}{(2N)!}}\sum_{k=0}^{\lfloor N/2\rfloor }\frac{1}{2^kk!\sqrt{(N-2k)!}}|k\rangle.
\label{eq:explicitGS}
\end{equation}
Hence both $|\!\operatorname{CBA}\rangle$ and $|\!\operatorname{TF}\rangle$ present approximately equal QFI
and a Heisenberg scaling $F_Q \propto N^2$. 
It is well known that the ground state of the TF phase approaches a TF state~\cite{ZhangPRL2013} 
and that the latter exhibits a QFI with Heisenberg scaling with $N$~\cite{HollandPRL1993, PezzeRMP, LuckeSCIENCE2011}.
For an analysis of the  QFI of the ground state of an antiferromagnetic spin-1 BEC see \cite{WuPRA2016}.
Conversely, Eq.~(\ref{QFICBA}) is a novel and less evident result. 
To gain some intuition regarding the large amount of useful entanglement found in the BA phase at $q=0$, let us rewrite the Hamiltonian~(\ref{eq:Hamiltonian}) in 
terms of the $\mathbf{\hat{S}}$ and $\mathbf{\hat{A}}$ operator manifolds of Eq.~(\ref{SAoperators}). 
We obtain, up to c-numbers,
\begin{equation} \label{eq:HamiltonianSA}
\hat H = 2\left(\lambda \hat S_x^2 - \frac{q}{3}\hat S_z\right) + 2\left(\lambda \hat A_y^2 - \frac{q}{3}\hat A_z\right),
\end{equation}
which is a sum of two (non-commuting) Lipkin-Meshkov-Glick Hamiltonians for $\mathbf{\hat S}$ and $\mathbf{\hat A}$, respectively. 
Since $\lambda<0$, the ground state of the first term in Eq.~(\ref{eq:HamiltonianSA}), 
$2 (\lambda \hat S_x^2 - \frac{q}{3}\hat S_z )$, at $q=0$ is a NOON state aligned along the $S_x$-axis
(i.\,e. a superposition of the maximum and the minimum eigenstates of $\hat{S}_x$). 
Its QFI saturates the Heisenberg limit for rotations generated by $\hat S_x$. 
Similarly, the ground state of the second term of Eq.~(\ref{eq:HamiltonianSA}), at $q=0$, is a NOON state aligned along the $A_y$-axis. This hints at large amounts of entanglement in the CBA state. However, since the symmetric and antisymmetric spin algebras share the same central mode $\hat{a}_0$ and therefore do not commute with each other, a more detailed inspection of the ground state is required. To this end, let us trace out the 
$\hat{h}$ mode. This leaves us with the state
\begin{equation}
\hat{\rho} = \sum_{N_h = 0}^N P(N_h) \vert \phi_{N_h} \rangle \langle \phi_{N_h} \vert, 
\end{equation}
in the modes $(\hat{a}_0,\hat{g})$, where $P(N_h)$ is the probability to measure $N_h$ particles in the $\hat{h}$ mode, 
and  $\vert \phi_{N_h} \rangle$ is a state of $N - N_h$ particles in $(\hat{a}_0,\hat{g})$.
In Figure~\ref{fig:2}(a) we plot $P(N_h)$ as a function of $N_h$.
The most probable value of $N_h$ is $N_h=0$, and $P(N_h)=0$ for odd values of $N_h$.
Since $\hat{S}_x$ commutes with $\hat N_h\equiv \hat h^\dagger \hat h$ and $\langle \operatorname{CBA}\!|\hat S_x|\!\operatorname{CBA}\rangle=0$, the QFI decomposes according~to
\begin{equation}
F_Q \big[|\!\operatorname{CBA}\rangle, \hat{S}_x\big] = \sum_{N_h=0}^N  P(N_h) F_Q \big[|\phi_{N_h}\rangle, \hat{S}_x\big].
\end{equation}
In Figure~\ref{fig:2}(b) we show $F_Q [|\phi_{N_h}\rangle, \hat{S}_x]$ as a function of $N_h$. 
Large values of the QFI are observed up to $N_h \simeq N/2$, in accordance with the presence of 
macroscopic superposition states~\cite{shortversion}. As can be seen from the Husimi distributions in Figure~\ref{fig:2}(c), for $N_h>N/2$ the $|\phi_{N_h}\rangle$ resemble NOON states along $\hat S_x$. 
This explains the Heisenberg scaling of the QFI~(\ref{QFICBA}).

\begin{figure}[t]
	\begin{center}
	\includegraphics[width=\columnwidth]{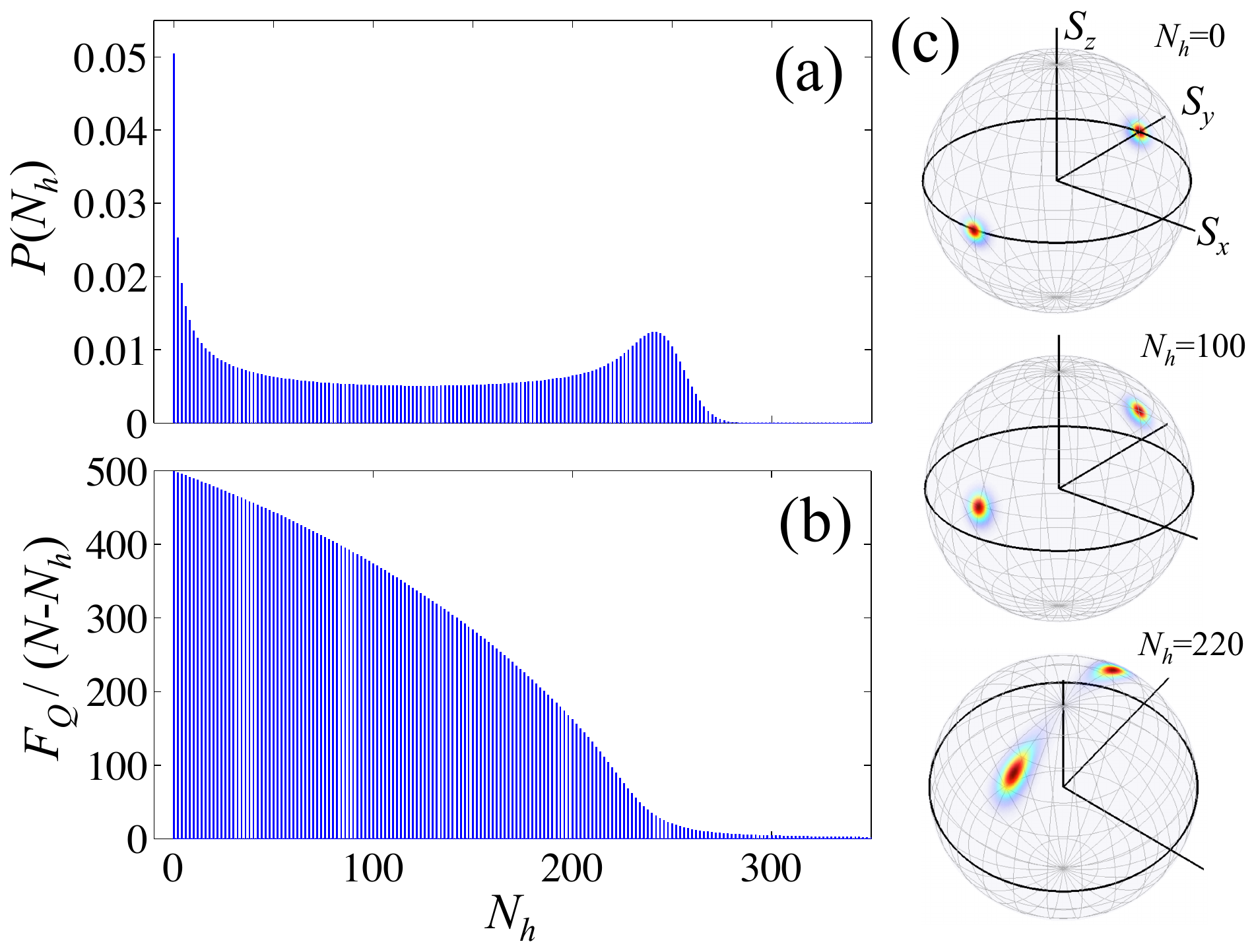}
	\end{center}
	\caption{(Color online) (a) Probability $P(N_h)$ to find $N_h$ particles in the $h$ 
	mode of $|\!\operatorname{CBA}\rangle$. 
	(b) QFI $F_Q [|\phi_{N_h}\rangle, \hat{S}_x ]$ of the states $|\phi_{N_h}\rangle$ 
	obtained by projecting $|\!\operatorname{CBA}\rangle$
	onto a fixed number of particles $N_h$ in the $h$ mode.
	(c) Husimi distributions of $|\phi_{N_h}\rangle$ for different values of $N_h$ showing 
	the presence of NOON-like states in a very broad range of $N_h$ values. In this Figure $N=500$.}
	\label{fig:2}
\end{figure}

We note that the $\mathbf{\hat{S}}$ manifold can be manipulated experimentally by radiofrequency pulses coupling the $m_f=0$ to the $m_f=\pm 1$ modes.
An atomic clock using transformations in the $\mathbf{\hat{S}}$ manifold of a spin-1 BEC has been demonstrated in Ref.~\cite{KrusePRL2016}, 
see also \cite{GrossNATURE2011, HamleyNATPHYS2012, PeiseNATCOMM2015} for squeezing of the $\mathbf{\hat{S}}$ spin.
Our results thus reveal the possibility to attain a sensitivity close to the HL preparing the spin-1 BEC in its ground state at $q=0$.
Since, when starting with the $m_f=0$ BEC, $q=0$ is reached after an adiabatic variation of $q$ that is half as large as the one required to arrive in the TF regime, implementing $|\!\operatorname{CBA}\rangle$ is less demanding in terms of BEC stability than the experiment reported in \cite{LuoSCIENCE2017}. 

Finally, in Appendix~\ref{app:GS_optMeas} we prove that a measurement of $(\hat{N}_+,\,\hat{N}_-)$ is, at any $\theta$, optimal for both $|\!\operatorname{CBA}\rangle$ and $|\!\operatorname{TF}\rangle$.
Optimal interferometric transformations $R_\mathrm{opt}^{(\operatorname{TF})}$ leave the TF state in the $N_0=0$ subspace, thus rendering $(\hat{N}_+,\,\hat{N}_-)$ equivalent to $\hat{D}$. 
A similar argument, see Appendix~\ref{app:GS_optMeas}, applies to $|\!\operatorname{CBA}\rangle$.  
Hence for both states and any phase~$\theta$ the experimentally relevant measurement of $\hat{D}$ turns out to be optimal. 


\section{Quasi-Adiabatic state preparation}

\begin{figure}[t!]
	\begin{center}
		\includegraphics[width=\columnwidth]{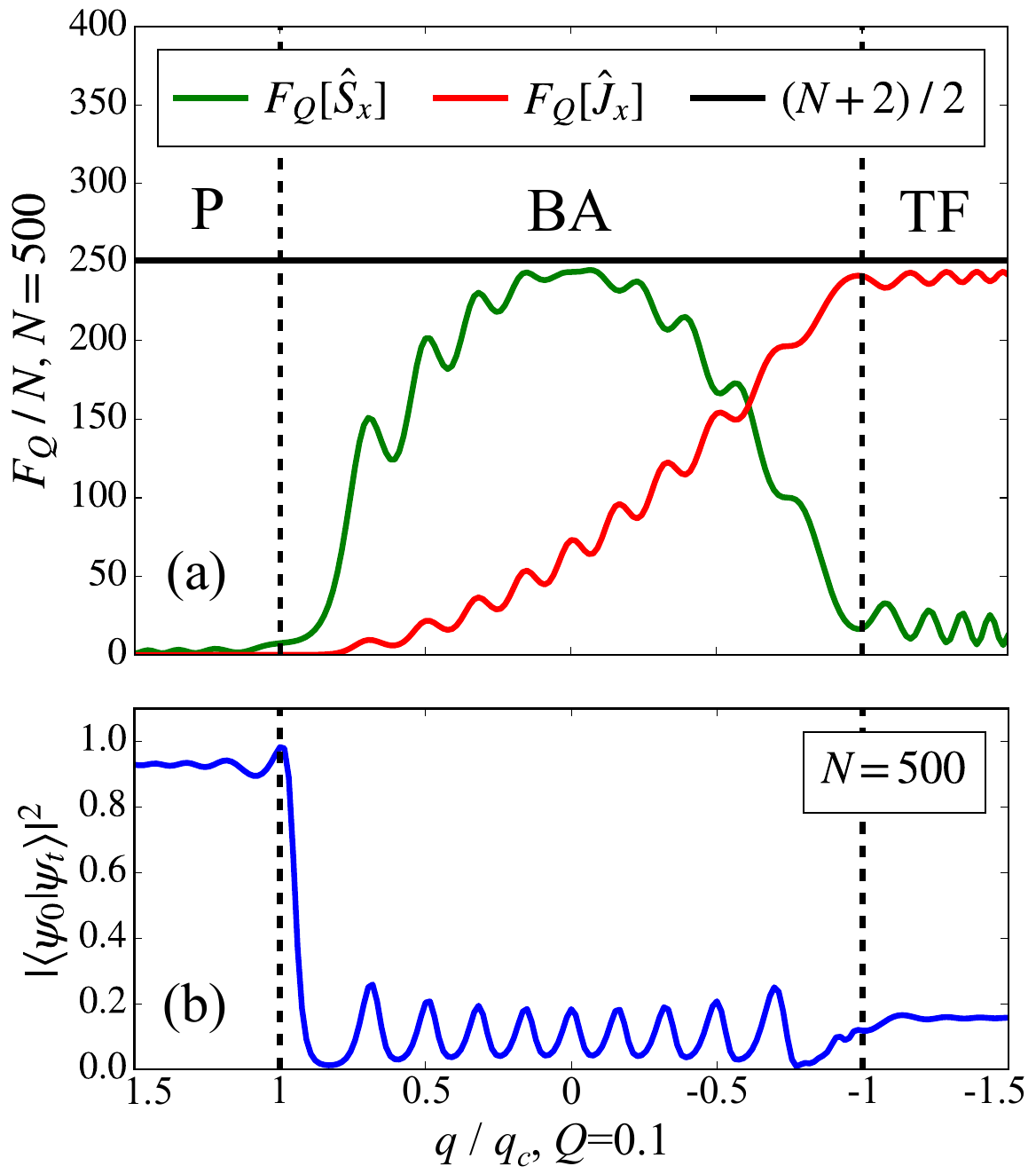}
	\end{center}
	\caption{(Color online) Quasi-adiabatic evolution of a ferromagnetic spin-1 BEC initialized in $m_f=0$. 
	The quadratic Zeeman shift $q(t)$ is linearly ramped with a slope $\propto Q$. 
	The quantum phase transitions are indicated by dashed vertical lines. (a) $F_Q/N$, where the 
	interferometric transformation is generated by $\hat{S}_x$~(green) and $\hat J_x$~(red), which 
	are optimal for $|\!\operatorname{CBA}\rangle$ and $|\!\operatorname{TF}\rangle$, respectively. 
	The solid horizontal line is $F_Q/N=(N+2)/2$, for comparison. 
	(b) Overlap $|\langle\psi(t)|\psi_0(q(t))\rangle|^2$ of the time evolved state 
	$|\psi(t)\rangle$ with the ground state at $q(t)$.}
	\label{fig:3}
\end{figure}

\begin{figure}[t]
	\begin{center}
		 \includegraphics[width=\columnwidth]{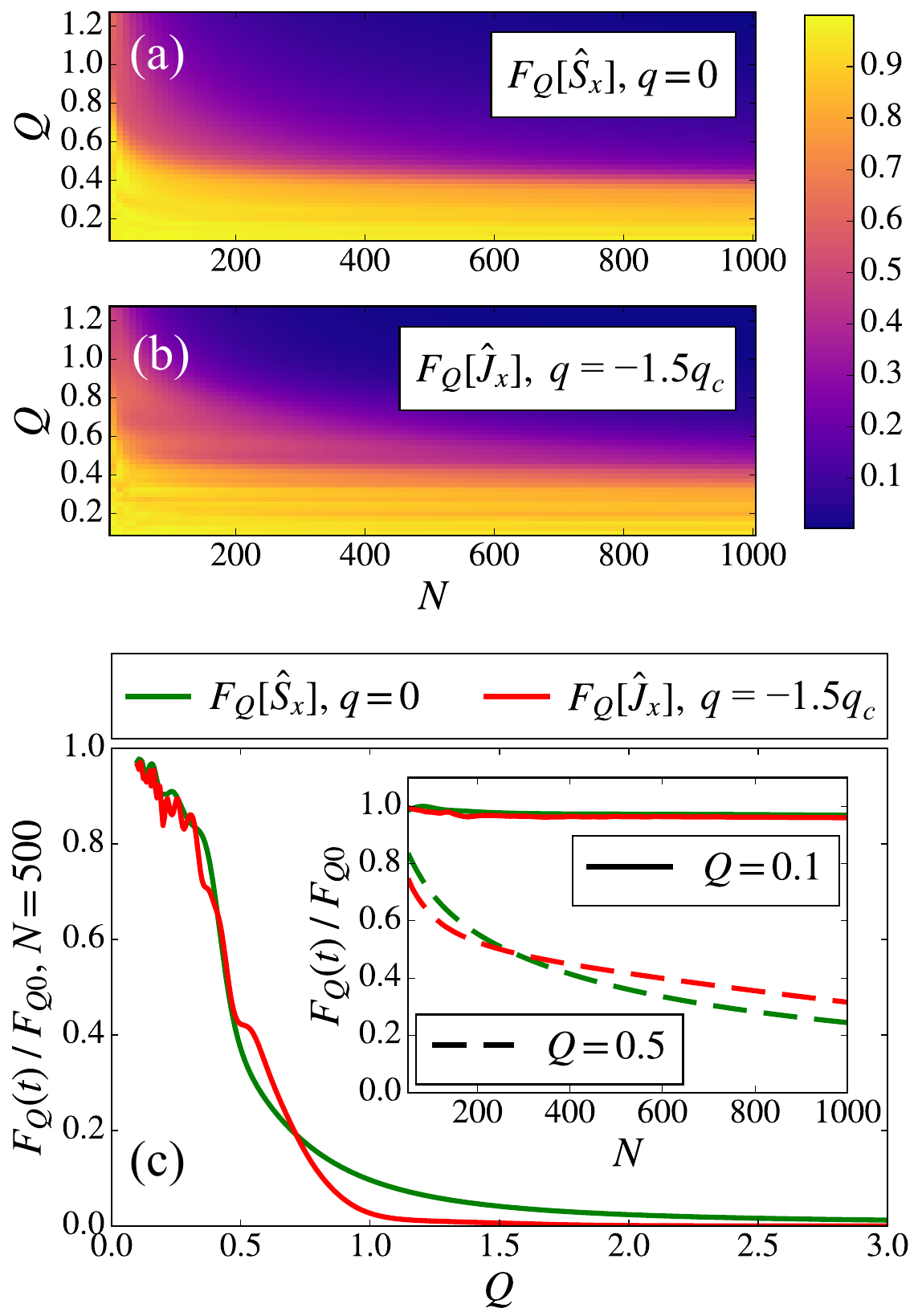}
	\end{center}
	\caption{(Color online) QFI of the states obtained via quasi-adiabatic evolution up to $q=0$ [panel (a)] and $q=-1.5 q_c$ [panel (b)] 
	normalized to the QFI of the respective ground states $|\!\operatorname{CBA}\rangle$ and $|\!\operatorname{TF}\rangle$, 
	indicated as $F_{Q0}$. 
	Panel (c) shows the normalized QFI as a function of $Q$. The green line corresponds to $q=0$, the red line to $q = -1.5 q_c$.
	The inset shows the normalized QFI as a function of $N$.
	In all panels, the interferometric transformations are generated by $\hat{S}_x$ and $\hat{J}_x$ 
	as is optimal for $|\!\operatorname{CBA}\rangle$ and $|\!\operatorname{TF}\rangle$, respectively.}
	\label{fig:4}
\end{figure}

Next we consider an experimental sequence for the variation of $q(t)$ as the one recently discussed in Ref.~\cite{LuoSCIENCE2017}. 
We assume a BEC prepared at $q/q_c=1.5$ in the P phase where $|\psi(t=0)\rangle=|k=0\rangle$. 
The value of the quadratic Zeeman term is varied following the ramp 
$q(t)/q_c=1.5 - q_c  Q t/4$, where $Q>0$ characterizes the non-adiabaticity of the process.
Figure~\ref{fig:3} illustrates our observations for $N=500$ particles and $Q=\num{0.1}$. 
We find that the QFI is hardly affected by the finite ramping speed. This is particularly striking since, as demonstrated 
in Figure~\ref{fig:3}(b), the fidelity $|\langle\psi(t)|\psi_0(q(t))\rangle|^2$ with the respective ground state $|\psi_0\rangle$ is dramatically diminished. 
The oscillations present in both Figure~\ref{fig:3}(a) and (b) resemble the ones found in 
Ref.~\cite{LuoSCIENCE2017} for the conversion efficiency $\langle \hat{N}_++\hat{N}_-\rangle/N$.

Note that at the critical points the energy gap between the ground and first excited state 
scales $\propto N^{-1/3}$~\cite{ZhangPRL2013}, and that a larger $Q$ means that the phase boundaries are crossed more rapidly.
Hence, enlarging $N$ or $Q$ displaces $|\psi(t)\rangle$ further away from the respective ground state, creating a larger number of excitations. 
Figure~\ref{fig:4} shows which fraction of the QFI for $|\!\operatorname{CBA}\rangle$ and $|\!\operatorname{TF}\rangle$ is accessible within finite time. As expected, it decreases with both $N$ and $Q$. Note that we vary $N$ at constant $q_c$. In Figure~\ref{fig:4}(b) we show slices through $F_Q(N,Q)$ at fixed $N$ or $Q$, respectively. The wavy distortions are due to the mentioned oscillations in the QFI, whose frequency depends both on $Q$ and---less pronounced---on $N$. We find large parts of the QFI conserved during non-adiabatic evolutions with $Q\lesssim 0.5$, and an overall rather small dependence on $N$. These features significantly ease experiments. Note particularly that at constant $Q$ the overall ramping time scales linearly with $N$, while the factual (dimensionful) speed of the linear ramp goes even as $\mathrm{d}q/\mathrm{d}t\propto N^2$. Together with Figure~\ref{fig:4}(a) this implies that enlarging the particle number reduces the requirements on adiabaticity and BEC stability. 
A pronounced dependence on whether the ramp of $q(t)$ is terminated at $q=0$ (CBA) or $q=-1.5 q_c$ (TF) 
is not discernible. This is consistent with the numerical analysis in~\cite{LuoSCIENCE2017} showing 
that the second phase transition---in contrast to the first one---has but little impact on the amount of created excitations.


\section{Finite measurement precision}

\begin{figure}[t]
	\begin{center}
		\includegraphics[width=\columnwidth]{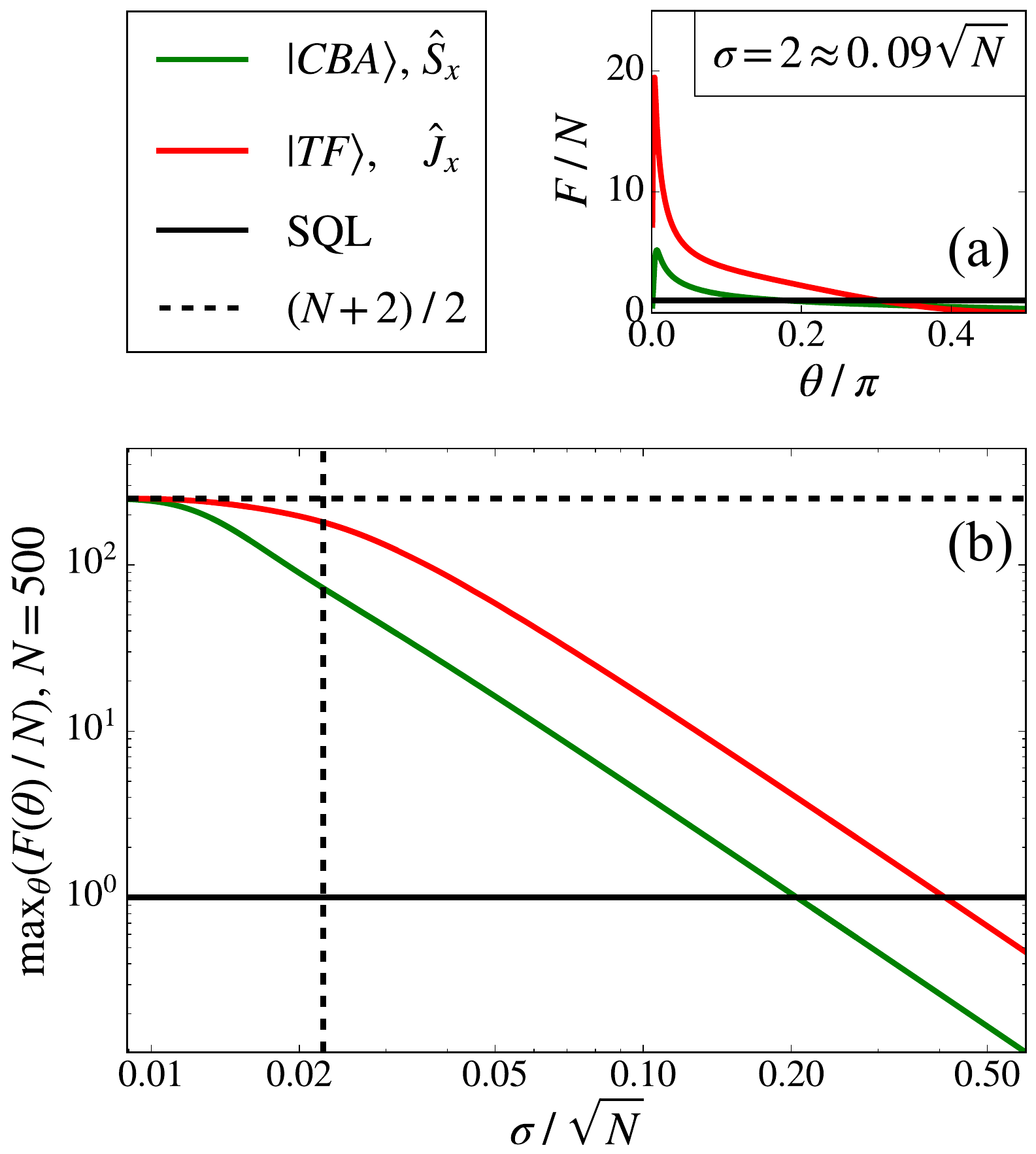}
	\end{center}
	\caption{(Color online) FI in the presence of Gaussian measurement uncertainty with variance 
	$\sigma^2$ for both $N_\pm$. (a) Dependence of the FI on the interferometric phase~$\theta$. 
	(b) Peak value of the FI as a function of $\sigma$. The dashed vertical line indicates measurement noise with a spread corresponding to a single particle ($\sigma=1/2$).
	In this Figure $N=500$.
	}
	\label{fig:5}
\end{figure}

To investigate the impact of a finite measurement resolution, we assume that the detection of both $\hat N_\pm$ is affected by Gaussian noise with variance $\sigma^2$, leading to an imprecise measurement of the number of particles. In this case, the actual measurement probabilities ensue from a convolution of the ideal quantum theoretical result with a Gaussian probability distribution $\Gamma_{\sqrt{2}\sigma}(x)$ of variance $2\sigma^2$ and zero mean. We thus determine the classical FI from the effective probability distribution
\begin{equation}
P_\mathrm{eff}(D|\theta)=\sum_{D'=-N}^N\! \Gamma_{\sqrt{2}\sigma}(|D-D'|)P(D'|\theta),
\end{equation}
where $P(D'|\theta)$ is the noiseless probability to find $D'$ upon measuring $\hat{D}$ at a phase $\theta$. Figure~\ref{fig:5} illustrates how the FI is affected by the detection uncertainty. Panel (a) shows for both $|\!\operatorname{CBA}\rangle$ and $|\!\operatorname{TF}\rangle$ that, while $F(\theta)$ as a whole is strongly damped, pronounced maxima at small $\theta$ remain far above the SQL, in close analogy to experiments presented in \cite{LuckeSCIENCE2011}. From panel (b) we infer that for worse than single particle detection $(\sigma=1/2)$ these peak values of the FI decrease approximately $\propto \sigma^{-2}$. Evaluating the relative standard deviation $\sigma_{\max}(N)/\sqrt{N}$ up to which the FI yet exceeds the SQL we have found that the TF state is slightly less sensitive to particle counting noise than the CBA state: $\sigma_{\max}[|\!\operatorname{TF}\rangle]/\sqrt{N} \approx \num{0.4}$ while $\sigma_{\max}[|\!\operatorname{CBA}\rangle]/\sqrt{N} \approx \num{0.2}$. Both $\sigma_{\max}/\sqrt{N}$ are easily undercut by state-of-the-art experiments \cite{LuckeSCIENCE2011}. 


\section{Parametric amplification}
\label{sec:PA}


\begin{figure}[t]
	\begin{center}
		\includegraphics[width=\columnwidth]{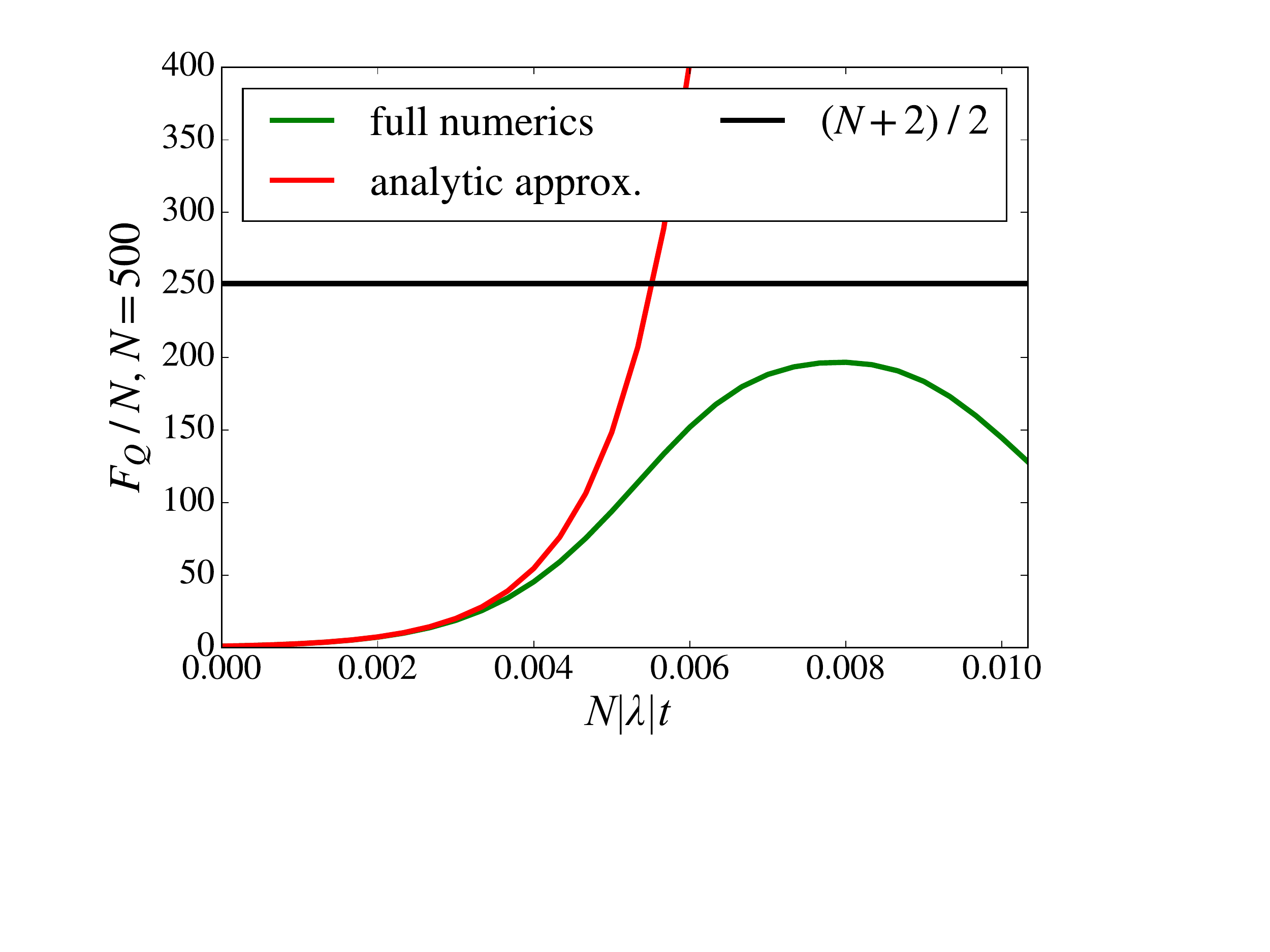}
	\end{center}
	\caption{(Color online) QFI attainable by a quench of the quadratic Zeeman shift to its resonance value $q_r$ and subsequent free evolution of the $m_f=0$ BEC. In the experimentally relevant regime~\cite{Oberthaler2017}, the analytic (red) approximation~\eqref{eq:QFI_PA} for small $t$ coincides with the full numeric simulation (green) of Hamiltonian~\eqref{eq:Hamiltonian}. Even at times significantly exceeding current technical feasibility the QFI fails to reach the level of the TF state indicated by the threshold (black).}
	\label{fig:6}
\end{figure}

Finally, we compare the \mbox{(quasi-)}adiabatic state preparation with the 
dynamical generation of entanglement following a quench of $q$.
Such a quench may render the initial $m_f=0$ condensate dynamically unstable. 
Spin-changing collisions populate $m_f=\pm 1$, thereby 
generating entanglement~\cite{UedaPR2012, DuanPRL2000, PuPRL2000, DuanPRA2002, GabbrielliPRL2015, SzigetiPRL2017}. 

Assuming $\langle \hat{N}_\pm\rangle \ll N$, in line with experiments \cite{KrusePRL2016, PeiseNATCOMM2015, HamleyNATPHYS2012, Oberthaler2017}, we may approximate $N_0\simeq N$ and $\hat{a}_0^\dagger,\,\hat{a}_0 \simeq \sqrt{N_0}$, which simplifies the Hamiltonian~\eqref{eq:Hamiltonian} to $\hat H=\hat H_g+\hat H_h$ with 
$\hat H_g=\alpha \hat g^\dag \hat g+ \beta(\hat g^2+\hat g^{\dagger 2})/2$ and 
$\hat H_h=\alpha \hat h^\dag \hat h- \beta (\hat h^2+\hat h^{\dagger 2})/2$, 
where $\alpha = q + \lambda (N-1/2)$ and $\beta=N\lambda$. 

As before, the generated entanglement can be used for interferometric transformations in the symmetric and anti-symmetric subspaces. The corresponding QFI of the state $|\psi(t)\rangle$ obtained by free evolution after a quench at $t=0$ reaches  
\begin{equation}
\frac{1}{N}F_Q[|\psi(t)\rangle,\,\hat R^{(g/h)}_\mathrm{opt}]\\
=1+2\left(\langle \hat{N}_\pm\rangle+\Delta \hat{N}_\pm\right)
\label{eq:QFI_PA}
\end{equation}
with 
$(\Delta \hat{N}_\pm)^2=\langle\hat{N}_\pm\rangle (\langle\hat{N}_\pm\rangle+1 )$, 
$\langle \hat{N}_\pm\rangle = (N \lambda)^2\tau^2\sinh^2(t/\tau)$, and $\tau = \left(\beta^2-\alpha^2\right)^{-1/2}$, where $\hat R^{(g/h)}_\mathrm{opt}$ is the respective generator which maximizes the QFI. This expression is valid only for short times, as long as the assumption $\langle \hat{N}_\pm\rangle \ll N$ holds. The explicit form of $\hat R^{(g/h)}_\mathrm{opt}$ and a derivation of Eq.~\eqref{eq:QFI_PA} are presented in Appendix~\ref{app:PA}.
Tuning $q$ allows to arbitrarily choose $\alpha$. 
At $\alpha=0$ the Hamiltonian is reduced to spin-changing collisions only. As expected, 
this affords the strongest growth of $\langle\hat{N}_\pm\rangle$ and thus of $F_Q$, 
entailing the definition of the resonance value $q_{r}=-\lambda(N-1/2)$. 
Applying Eq.~\eqref{eq:QFI_PA} to 
recent spin squeezing experiments \cite{KrusePRL2016, PeiseNATCOMM2015, HamleyNATPHYS2012} provides a relative QFI, $F_Q/N$, which ranges from \num{5} to \num{20}, thereby corresponding to a $F_Q/N^2$ of less than \num{2e-2} only. Recall that, in the ideal case, quasi-adiabatic entanglement generation as discussed in this paper allows for $F_Q/N^2\approx0.5$.

In the absence of technical noise it would be advantageous to extend parametric amplification protocols to longer evolution times.
We therefore numerically simulate the anticipated further evolution of the QFI under the full Hamiltonian~\eqref{eq:Hamiltonian} at $q=q_r$. Figure~\ref{fig:6} indicates that even under ideal experimental conditions parametric amplification is unable to reach the QFI attainable by the quasi-adiabatic approach.
In the best scenario, dynamical spin-changing collision creates entangled states with a QFI $F_Q/N^2\approx 0.257$ \cite{WuPRA2016}.


\section{Conclusions}
\label{sec:Conclusions}

We have studied the generation of entanglement useful for quantum enhanced interferometry with ferromagnetic spin-$1$ BECs, focusing 
on the experimentally relevant case of quasi-adiabatic driving through quantum phase transitions. We have shown that, starting out from a BEC in $m_f=0$ and a quadratic Zeeman shift of $q> q_c$, at $q=0$ and thus halfway to the TF state another highly entangled state, 
$|\!\operatorname{CBA}\rangle$, of equal interferometric value (approximately equal QFI) emerges. This allows us to propose an alternative interferometric scheme admitting Heisenberg scaling.

For both $|\mathrm{TF}\rangle$ and $|\!\operatorname{CBA}\rangle$, optimal values of the QFI are obtained with interferometric transformations corresponding to common radio-frequency coupling techniques. The optimal measurement procedure is based on the well established counting of particles in the $m_f=\pm 1$ modes. According to our findings surpassing the standard quantum limit is expected to remain feasible under realistic conditions, when the quasi-adiabatic transition is performed at finite speed, and measurement uncertainty is present. While the TF state is less sensitive to imperfections of atom counting, the CBA state has the advantage of being quasi-adiabatically reachable within half the time. Both regimes favorably compare to squeezing through parametric amplification, thus constituting a promising source of interferometrically useful entanglement.


\begin{acknowledgments}
We acknowledge support by the SFB 1227 ``DQ-mat'', projects A02 and B01, of the 
German Research Foundation (DFG). M.\,G. thanks the Alexander von Humboldt foundation for support.
\end{acknowledgments}


\appendix

\section{Adiabatic phase transition}
\label{app:GS}

\subsection{Gell-Mann Covariance Matrix}
\label{app:GS_Cov}

The collective Gell-Mann operators read
\begin{equation}
\begin{alignedat}{3}
\hat{G}_1 &= \frac{\hat{a}_{-1}^\dagger \hat{a}_0 + \hat{a}_0^\dagger \hat{a}_{-1}}{2},\hspace{1.5em}
&\hat{G}_2 &= \frac{\hat{a}_{-1}^\dagger \hat{a}_0 - \hat{a}_0^\dagger \hat{a}_{-1}}{2i},\\
\hat{G}_3 &=\frac{\hat{a}_{-1}^\dagger \hat{a}_{-1}-\hat{a}_0^\dagger \hat{a}_0}{2},\hspace{1.5em}
&\hat{G}_4 &=\frac{\hat{a}_1^\dagger \hat{a}_{-1}+\hat{a}_{-1}^\dagger \hat{a}_1}{2}, \\
\hat{G}_5 &= \frac{\hat{a}_{-1}^\dagger \hat{a}_1-\hat{a}_1^\dagger \hat{a}_{-1}}{2i},\hspace{1.5em}
&\hat{G}_6 &=\frac{\hat{a}_0^\dagger \hat{a}_1+\hat{a}_1^\dagger \hat{a}_0}{2},\\
\hat{G}_7 &=\frac{\hat{a}_0^\dagger \hat{a}_1-\hat{a}_1^\dagger \hat{a}_0}{2i},\hspace{1.5em}
&\hat{G}_8 &=\frac{\hat{a}_{-1}^\dagger \hat{a}_{-1}+\hat{a}_0^\dagger \hat{a}_0 - 2\hat{a}_1^\dagger \hat{a}_1}{2\sqrt{3}}.
\end{alignedat}
\end{equation}
Consider an arbitrary (normalized) state $|\psi\rangle=\sum_k c_k |k\rangle$ with $\hat{D}|\psi\rangle=0$. The covariance matrix of $\hat{\mathbf{G}}$ is block diagonal,
\begin{align}
\boldsymbol{\Gamma}_{\hat{\mathbf{G}}}&=\boldsymbol{\Gamma}_{(\hat G_1, \hat G_2, \hat G_6, \hat G_7)}\oplus \boldsymbol{\Gamma}_{(\hat G_3,\hat G_8)}\oplus (\Delta\hat{G}_4)^2\oplus (\Delta\hat{G}_5)^2
\end{align}
with $(\Delta \hat G_i)^2$ the variance of $\hat G_i$.
The twofold degenerate eigenvalues of $\boldsymbol{\Gamma}_{(\hat G_1, \hat G_2, \hat G_6, \hat G_7)}$ are
\begin{equation}
\lambda_\pm= A \pm |B|
\end{equation}
with
\begin{align}
\begin{split}
A &= \frac{1}{4}\left(N+\sum_{k=0}^{\lfloor N/2\rfloor }|c_k|^2k(2N-4k-1)\right),\\
B &= \frac{1}{2}\sum_{k=0}^{\lfloor N/2-1 \rfloor}c_k^*c_{k+1}(k+1)\sqrt{(N-2k)(N-2k-1)}
\end{split}
\end{align}
and corresponding eigenvectors
\begin{equation}
\begin{alignedat}{5}
\mathbf{u}_+^{(1)}&=\frac{1}{\sqrt{2}|B|}(&&\Im (B),&&\Re (B),&0,&&\,|B|),\\
\mathbf{u}_+^{(2)}&=\frac{1}{\sqrt{2}|B|}(&&\Re (B),-&&\Im (B),&\,|B|,&&0),\\
\mathbf{u}_-^{(1)}&=\frac{1}{\sqrt{2}|B|}(-&&\Im (B),-&&\Re (B),&0,&&\,|B|),\\
\mathbf{u}_-^{(2)}&=\frac{1}{\sqrt{2}|B|}(-&&\Re (B),&&\Im (B),&\,|B|,&&0),
\end{alignedat}
\end{equation}
where $\Re$ and $\Im$ denote the real and imaginary part, respectively.
Note that if $c_k\in\mathbb{R}$ for all $k$ we obtain $\Im(B)=0$ and hence $\boldsymbol{\Gamma}_{(\hat G_1, \hat G_2, \hat G_6, \hat G_7)}=\boldsymbol{\Gamma}_{(\hat G_1, \hat G_6)}\oplus \boldsymbol{\Gamma}_{(\hat G_2,\hat G_7)}$. The eigenvalues of $\boldsymbol{\Gamma}_{(\hat G_3,\hat G_8)}$ are
\begin{equation}
\lambda_0=0,\;\;\lambda_1 = 3\left(\sum_{k=0}^{\lfloor N/2\rfloor }|c_k|^2k^2-\left(\sum_{k=0}^{\lfloor N/2\rfloor}|c_k|^2k\right)^2\right).
\end{equation}
The corresponding eigenvectors read
\begin{equation}
\mathbf{u}_0=\left(1,\sqrt{3}\right),\;\;\mathbf{u}_1=\left(\sqrt{3},-1\right).
\end{equation}
Finally,
\begin{equation}
(\Delta \hat G_4)^2=(\Delta\hat G_5)^2=\frac{1}{2}\sum_{k=0}^{\lfloor N/2\rfloor }|c_k|^2k(k+1).
\end{equation}

Numerically evaluating $\boldsymbol{\Gamma}_{\hat{\mathbf{G}}}$ in the ground state $|\psi_0(q)\rangle$ of the Hamiltonian~\eqref{eq:Hamiltonian} we find that only the (two-fold degenerate) eigenvalues $\lambda_+$ and $(\Delta \hat G_4)^2=(\Delta\hat G_5)^2$, depicted in Figure~\ref{fig:1} and Figure~\ref{fig:3}(a), significantly differ from zero for large $N$. They correspond to the eigenoperators $\hat{S}_x=(\hat{G}_{1}+\hat{G}_{6})/\sqrt{2}$, $\hat{A}_y=(\hat{G}_{2}+\hat{G}_{7})/\sqrt{2}$, which become optimal in the BA phase, and $\hat{J}_x=\hat{G}_{4}$, $\hat{J}_y=-\hat{G}_{5}$, prevalent in the TF phase, respectively.

\subsection{Properties of the CBA state}
In this section, we provide some analytical results on the CBA state, i.\,e., the ground state at $q=0$.
\label{app:GS_q0}
\subsubsection{Coefficients in the Fock basis}
Let us introduce the collective pseudospin-$1$ operator $\mathbf{\hat{L}}$ composed of
\begin{equation}
\begin{alignedat}{2}
\hat{L}_x &= \frac{1}{\sqrt{2}}(\hat{a}_0^\dagger\hat{a}_1 + \hat{a}_0^\dagger\hat{a}_{-1} + \hat{a}_1^\dagger\hat{a}_0 + \hat{a}_{-1}^\dagger\hat{a}_0) && =  \hphantom{-}2\hat{S}_x,\\ 
\hat{L}_y &= \frac{1}{i\sqrt{2}}(\hat{a}_0^\dagger\hat{a}_1 -\hat{a}_0^\dagger\hat{a}_{-1} -\hat{a}_1^\dagger\hat{a}_0 + \hat{a}_{-1}^\dagger\hat{a}_0 ) && =  \hphantom{-}2\hat{A}_y,\\
\hat{L}_z &=\hat{a}_{-1}^\dagger\hat{a}_{-1} -\hat{a}_1^\dagger\hat{a}_1 && =  -2\hat{J}_z.
\end{alignedat}
\end{equation}
This allows us to express the ground state of the Hamiltonian~\eqref{eq:Hamiltonian} in the subspace of $D=0$ at $q=0$ as \cite{LawPRL1998}
\begin{align}
\begin{split}
|\!\operatorname{CBA}\rangle &= \frac{1}{\sqrt{(2N)!}}\left(\hat{L}_+\right)^N|N_-=0,\,N_0=0,\,N_+=N\rangle,\\
\hat{L}_+ &= \hat{L}_x + i\hat{L}_y = \sqrt{2}(\hat{a}_{-1}^\dagger \hat{a}_0 + \hat{a}_0^\dagger \hat{a}_1).
\label{eq:algebraicGS}
\end{split}
\end{align}
We expand the operator $\hat{L}_+^N$, leading to
\begin{align}\label{eq:expandedGS}
|\!\operatorname{CBA}\rangle &= \sqrt{\frac{2^N}{(2N)!}}\sum_{k=0}^{\lfloor N/2\rfloor }\mathfrak{C}_k[\hat a_0, \hat a_0^\dagger]\hat a_{-1}^{\dagger k}\hat a_1^{N-k}|0,\,0,\,N\rangle\notag\\
&=\sqrt{\frac{2^NN!}{(2N)!}}\sum_{k=0}^{\lfloor N/2\rfloor}\mathfrak{C}_k[\hat a_0, \hat a_0^\dagger]|k,\,0,\,k\rangle,
\end{align}
where $\mathfrak{C}_k[\hat a, \hat a^\dagger]$ denotes the sum over all possible products composed of $k$ operators $\hat{a}$ and $N-k$ operators $\hat{a}^\dagger$ in arbitrary order, and we have used that $\hat a_{-1}^{\dagger}$ and $\hat{a}_1$ commute with each other as well as with $\hat{a}_0$ and $\hat{a}_0^\dagger$. Each term in $\mathfrak{C}_k[\hat a_0, \hat a^\dagger_0]$ leads to a total creation of $N-2k$ particles in mode $\hat a_0$. Since we apply $\hat L_+^N$ to a state containing zero particles in the central mode, only values of $k\leq N/2$ contribute to Eq.~\eqref{eq:expandedGS}. An explicit evaluation can be performed by means of the following
\begin{lemma}
	For an arbitrary bosonic mode with creation operator $\hat{a}^\dagger$ and vacuum $|0\rangle$
	\begin{align}
	\mathfrak{C}_k[\hat a, \hat a^\dagger]|0\rangle=\frac{N!}{2^k k! (N-2k)!}(\hat{a}^\dagger)^{N-2k}|0\rangle\;\;\;\forall k\leq N/2.
	\end{align}
	\label{lemma:1}
\end{lemma}
\begin{proof} 
	Since each term in the sum described by $\mathfrak{C}_k[\hat a, \hat a^\dagger]$ describes the creation of $N-2k$ particles, we can write
	\begin{equation}
	\mathfrak{C}_k[\hat a, \hat a^\dagger]|0\rangle=X(k)(\hat{a}^\dagger)^{N-2k}|0\rangle,
	\end{equation}
	which reduces the problem to the identification of the combinatorial factor $X(k)$.
	We use Wick's theorem \cite{Bogoliubov1959}, $\hat{a}|0\rangle=0$, and the fundamental Wick contractions 
	\begin{equation}
	\begin{alignedat}{2}
	&\mathrlap{\ssmile{1.1}}\hat{a}\,\hat{a}=\mathrlap{\ssmile{1.35}}\hat{a}^\dagger\hat{a}^\dagger=\mathrlap{\ssmile{1.25}}\hat{a}^\dagger\hat{a}&&=0,\\
	&\mathrlap{\ssmile{1.1}}\hat{a}\,\hat{a}^\dagger&&=1
	\end{alignedat}
	\end{equation}
	where $\mathrlap{\ssmile{1.6}}\hat{A}\,\hat{B}\equiv \hat A \hat B - \!:\!\!\hat A \hat B\!\!:$ and the double dots denote normal ordering.
	The contribution of each permutation to $X(k)$ is the number of variants it admits for enclosing all $k$ annihilation operators into $\hat{a}\hat{a}^\dagger$-contractions. Taking into consideration all permutations of $\hat{a},\,\hat{a}^\dagger$ reveals that $X(k)$ is the number of possibilities to tag $k$ unsorted disjoint tuples in a set of $N$ elements. There are $\binom{N}{N-2k}$
	different choices for the positions of the $\hat{a}^\dagger$ which are not going to be contracted. Thus, we merely have to count the number of possibilities to pair $2k$ objects. First arbitrarily arranging them and then compensating for the ordering of and within pairs we arrive at $(2k)!/(2^kk!)$.
	This completes the proof, since
	\begin{equation}
	X(k) = \binom{N}{N-2k}\frac{(2k)!}{2^kk!} = \frac{N!}{2^k k! (N-2k)!}.
	\end{equation}
\end{proof}

Applying Lemma~\ref{lemma:1} to Eq.~(\ref{eq:expandedGS}) gives
\begin{align}\label{eq:CBAfull}
|\!\operatorname{CBA}\rangle &= \sqrt{\frac{2^NN!^3}{(2N)!}}\sum_{k=0}^{\lfloor N/2\rfloor}\frac{1}{2^k k! \sqrt{(N-2k)!}}|k,\,N-2k,\,k\rangle
\end{align}
as reported in Eq.~\eqref{eq:explicitGS}.

\subsubsection{Quantum Fisher information}
As discussed in the main text, see Figure~\ref{fig:1}, and in Appendix~\ref{app:GS_Cov}, the QFI of $|\!\operatorname{CBA}\rangle$ is maximized by any $\hat{R}_\mathrm{opt}\in\operatorname{span}\{\hat S_x,\,\hat A_y\}$. We consider without loss of generality $\hat R_\mathrm{opt}=\hat S_x$. Then
\begin{align}
&F_Q[|\!\operatorname{CBA}\rangle,\,\hat{R}_\mathrm{opt}]=4(\Delta\hat{S}_x)^2\notag\\&=\sum_{k=0}^{\lfloor (N-1)/2\rfloor }(k+1)\left|\sqrt{N-2k}\;c_k + \sqrt{N-2k-1}\;c_{k+1}\right|^2,
\end{align}
where
\begin{equation}\label{eq:ck}
c_k=\sqrt{\frac{2^NN!^3}{(2N)!}}\frac{1}{2^k k! \sqrt{(N-2k)!}}
\end{equation}
are the Fock-state coefficients of $|\!\operatorname{CBA}\rangle$ from Eq.~(\ref{eq:CBAfull}), and $c_{k>N/2}\equiv 0$. Thus
\begin{align}
\begin{split}
&\sqrt{N-2k}\;c_k + \sqrt{N-2k-1}\;c_{k+1}\\&=\sqrt{\frac{2^NN!^3}{(2N)!}}\frac{N+1}{2^{k+1}(k+1)!\sqrt{(N-2k-1)!}},
\end{split}
\end{align}
which, after some rearrangements, leads to
\begin{align}
\begin{split}
&F_Q[|\!\operatorname{CBA}\rangle,\,\hat{R}_\mathrm{opt}]\\
&=\frac{(N+1)!^2}{2(2N)!}\sum_{k=0}^{\lfloor(N-1)/2\rfloor}\!\!\binom{N}{k}\binom{N-k}{k+1}2^{N-2k-1}.
\end{split}
\end{align}
The sum ensues from the following
\begin{lemma}
\begin{equation}
\sum_{k=0}^{\lfloor n/2\rfloor}\binom{m}{k}\binom{m-k}{n-2k}2^{n-2k}=\binom{2m}{n}\;\;\forall n\leq 2m
\label{eq:lemma2}
\end{equation}
\label{lemma:2}
\end{lemma}
\begin{proof}
Consider $2m$ sites grouped into $m$ pairs. For each $k$ the left-hand side of Eq.~\eqref{eq:lemma2}
is the number of possibilities to distribute $2k + (n - 2k) = n$ indiscernible objects on these $2m$ sites in such a way that exactly $k$ pairs are completed. The number of obtained pairs can assume values between $\max\{0,\,n-m\}$ and $\lfloor n/2 \rfloor$. Note that $k < n-m$ do not contribute to Eq.~\eqref{eq:lemma2}. Thus, summing over all $0\leq k\leq n/2$ amounts to counting the variants of distributing $n$ identical elements on $2m$ sites, which gives $\binom{2m}{n}$.
\end{proof}

Choosing $n=N-1$ and $m=N$ we obtain
\begin{align}
F_Q[|\!\operatorname{CBA}\rangle,\,\hat{R}_\mathrm{opt}]=N(N+1)/2.
\end{align}

\subsubsection{Mean particle number}
We further determine the exact mean particle numbers in the $\hat{a}_{\pm 1}$ modes, using
\begin{align}
\langle \hat{N}_\pm \rangle = \sum_k k|c_k|^2,
\end{align}
with the coefficients of the CBA state defined in Eq.~(\ref{eq:ck}). Applying Lemma~\ref{lemma:2} with $n=N$ and $m=N-1$, we find
\begin{align}
\langle \hat{N}_\pm\rangle &= \frac{N}{4}\cdot\frac{2N-2}{2N-1}\notag\\
&= \frac{N}{4}\left(1-\frac{1}{2N}+\mathcal{O}\left(\frac{1}{N^2}\right)\right).
\end{align}
For $N\gg 1$, we recover the mean field expression $\langle \hat{N}_\pm\rangle =N/4$.

\subsection{Optimal measurements}
\label{app:GS_optMeas}

We consider the two interferometrically relevant states $|\psi\rangle\in\{|\!\operatorname{CBA}\rangle, |\!\operatorname{TF}\rangle\}$ along with the respective optimal generators of the interferometric rotation
\begin{equation}
\begin{alignedat}{2}
&\hat{R}_\mathrm{opt}^{(\operatorname{CBA})}(\phi) &&= \frac{1}{2\sqrt{2}}\left(\operatorname{e}^{-i\phi}\hat{a}_0^\dagger\hat{a}_1 + \operatorname{e}^{i\phi}\hat{a}_0^\dagger\hat{a}_{-1} \right. \\
& &&\hphantom{=}\; \left. + \operatorname{e}^{i\phi}\hat{a}_1^\dagger\hat{a}_0 + \operatorname{e}^{-i\phi}\hat{a}_{-1}^\dagger\hat{a}_0\right),\\
&\hat{R}_\mathrm{opt}^{(\operatorname{TF})}(\phi)&& = \frac{1}{2}\left(\operatorname{e}^{-i\phi}\hat{a}_1^\dagger\hat{a}_{-1} + \operatorname{e}^{i\phi}\hat{a}_{-1}^\dagger\hat{a}_{1}\right)
\end{alignedat}
\end{equation}
providing $|\psi(\theta)\rangle=\operatorname{e}^{-i\theta\hat R_{\mathrm{opt}}}|\psi\rangle$. Let us first focus on $\phi=0$, where $\hat{R}_\mathrm{opt}^{(\operatorname{CBA})}(0)=\hat{S}_x$ and $\hat{R}_\mathrm{opt}^{(\operatorname{TF})}(0)=\hat{J}_x$,
and show that a measurement of $(\hat{N}_+,\,\hat{N}_-)$ is optimal at any~$\theta$. 

The projections $\hat{P}_{N_+,\,N_-}\equiv|N_-,\,N_0,\,N_+\rangle \langle N_-,\,N_0,\,N_+|$ with $N_0=N-N_+-N_-$
onto the eigenstates of $(\hat N_+, \hat N_-)$ are one-dimensional. In addition, $\langle\psi|\operatorname{d}\!\psi\rangle = 0$ and hence $\langle\psi(\theta)|\operatorname{d}\!\psi(\theta)\rangle = 0$ for all $\theta$. Therefore
\begin{equation}
\Re\langle\psi(\theta)|\hat{P}_{N_+,\,N_-}\hat{R}_\mathrm{opt}|\psi(\theta)\rangle=0 \;\; \forall N_+,N_-,\theta
\label{eq:optmeas}
\end{equation}
provides both the necessary and sufficient condition for optimality~\cite{Caves1994}. We observe that $\hat{R}_\mathrm{opt}^n|\psi\rangle$, $n\in\mathbb{N}_0$ has only real coefficients in the Fock basis $\{|N_-,\,N_0,\,N_+\rangle\}$. Furthermore,
\begin{equation}
\begin{alignedat}{2}
&\sum_{D=2n\hphantom{+1}}\hat{P}_D\hat{S}_x^{2m+1}&&|\!\operatorname{CBA}\rangle=0,\\
&\sum_{D=2n+1}\hat{P}_D\hat{S}_x^{2m}&&|\!\operatorname{CBA}\rangle=0
\end{alignedat}
\end{equation}
and
\begin{equation}
\begin{alignedat}{2}
&\sum_{N_\pm=2n\hphantom{+1}}\hat{P}_{N_\pm}\hat{J}_x^{2m+1}&&|\!\operatorname{TF}\rangle=0,\\
&\sum_{N_\pm=2n+1}\hat{P}_{N_\pm}\hat{J}_x^{2m}&&|\!\operatorname{TF}\rangle=0
\end{alignedat}
\end{equation}
for all $n,m\in\mathbb{N}_0$.
This entails
\begin{align}
\begin{split}
&\langle\psi|\hat{R}_\mathrm{opt}^{2n}\hat{P}_{N_+,\,N_-}\hat{R}_\mathrm{opt}^{2m+1}|\psi\rangle=0,\\
&\langle\psi|\hat{R}_\mathrm{opt}^{2m+1}\hat{P}_{N_+,\,N_-}\hat{R}_\mathrm{opt}^{2n}|\psi\rangle=0.
\end{split}
\end{align}
Thus
\begin{align}
\begin{split}
&\langle\psi(\theta)|\hat{P}_{N_+,\,N_-}\hat{R}_\mathrm{opt}|\psi(\theta)\rangle\\
&= i\sum_{j,l=0}^{\infty}\frac{(-1)^{j+l}}{(2j)!(2l)!}\theta^{2(j+l)+1}\\
&\hspace{3.8em}\left(\frac{1}{2j+1}\langle\psi|\hat{R}_\mathrm{opt}^{2j+1}\hat{P}_{N_+,\,N_-}\hat{R}_\mathrm{opt}^{2l+1}|\psi\rangle\right.\\ 
&\hspace{3.8em}\left.- \frac{1}{2l+1}\langle\psi|\hat{R}_\mathrm{opt}^{2j}\hat{P}_{N_+,\,N_-}\hat{R}_\mathrm{opt}^{2l+2}|\psi\rangle\right)\\
&\in i\mathbb{R},
\end{split}
\end{align}
which implies Eq.~\eqref{eq:optmeas}.

Next, we consider a measurement of $\hat{D}$, still for $\phi=0$. Due to $\hat{P}_D = \sum_{N_+-N_-=D} \hat{P}_{N_+,\,N_-}$,
Eq.~\eqref{eq:optmeas} holds also when $\hat{P}_{N_+,\,N_-}$ is substituted by $\hat{P}_D$. Since the $\hat{P}_D$ are no longer one-dimensional, this is not sufficient for optimality \cite{Caves1994}. 
However, we are able to show that for both $|\!\operatorname{CBA}\rangle$ and $|\!\operatorname{TF}\rangle$ the Hilbert space $\mathcal{H}=\operatorname{span}\{|N_+,\,N_0,\,N_-\rangle\}$ can be restricted to $\mathcal{H}'$ such that $\{\hat{R}_\mathrm{opt}^n|\psi\rangle, \,n\in\mathbb{N}_0\}\in\mathcal{H}'$ while the dimensionality of $\hat{P}_D\mathcal{H}'$ for any $D$ is one. Let us start with the TF state. Since $[\hat{J}_x,\,\hat{N}_0]=0$ 
and $\hat{N}_0|\!\operatorname{TF}\rangle=0$, we can choose $\mathcal{H}'=\hat{P}_{N_0=0}\mathcal{H}$. Regarding $|\!\operatorname{CBA}\rangle$, recall that $\hat{S}_x=\hat{L}_x/2$. Hence $[\hat{S}_x,\,\mathbf{\hat{L}}^2]=0$. 
Then $\mathbf{\hat{L}}^2|\!\operatorname{CBA}\rangle=N(N+1)|\!\operatorname{CBA}\rangle$~\cite{LawPRL1998} suggests to set $\mathcal{H}'=\hat{P}_{\mathbf{\hat{L}}^2=N(N+1)}\mathcal{H}$.
$\hat{L}_z$ has a non-degenerate spectrum in $\mathcal{H}'$. Thus $\hat{L}_z=-\hat{D}$ establishes the one-dimensionality of all $\hat{P}_D\mathcal{H}'$. 

To proceed to arbitrary $\phi$ we note that
\begin{equation}
\begin{alignedat}{3}
&\hat{R}_\mathrm{opt}^{(\operatorname{CBA})}(\phi)&&=\operatorname{e}^{2i\phi\hat J_z}&&\hat{S}_x\operatorname{e}^{-2i\phi\hat J_z},\\
&\hat{R}_\mathrm{opt}^{(\operatorname{TF})}(\phi)&&=\operatorname{e}^{i\phi\hat J_z}&&\hat{J}_x\operatorname{e}^{-i\phi\hat J_z}
\end{alignedat}
\end{equation}
and recall that the classical FI~\eqref{eq:CFI} depends only on $P(\mu|\theta)=\langle\psi(\theta)|\hat P_\mu|\psi(\theta)\rangle$ with $\mu$ the possible measurement outcomes. Because $\operatorname{e}^{-i\phi\hat J_z}\hat P_{N_+,\,N_-}\operatorname{e}^{i\phi\hat J_z}=\hat P_{N_+,\,N_-}$ and, thanks to $\hat{D}|\psi\rangle=0$, $\operatorname{e}^{-i\phi\hat J_z}|\psi\rangle=|\psi\rangle$, our results hold for any $\phi$.

\section{Parametric amplification}
\label{app:PA}

We start by introducing the generators of $\mathfrak{su}(1,1)$
\begin{equation}
\hat{L}_0^{(a)}=\frac{1}{4}(2\hat{a}^\dagger\hat{a}+1),\;\;\hat{L}_+^{(a)}=\frac{1}{2}\hat{a}^{\dagger 2},\;\;\hat{L}_-^{(a)}=\frac{1}{2}\hat{a}^2,
\end{equation}
where $\hat{a}^\dagger$ denotes some bosonic creation operator. Under the approximations $N_0\simeq N$ and $\hat{a}_0^\dagger,\,\hat{a}_0 \simeq \sqrt{N_0}$ the Hamiltonian~\eqref{eq:Hamiltonian} becomes, discarding c-numbers, $\hat{H} = \hat{H}_g+\hat{H}_h$ with
\begin{align}
\begin{split}
\hat{H}_{g/h} &=2\hbar\alpha \hat{L}_0^{(g/h)} + \hbar\epsilon_{g/h}\beta \left(\hat{L}_-^{(g/h)}+\hat{L}_+^{(g/h)}\right)
\end{split}
\end{align}
and $\epsilon_{g/h}=\pm 1$. Let us denote the eigenstates of $\hat{N}_g\equiv \hat g^\dag\hat g$ and $\hat N_h\equiv \hat h^\dag \hat h$ by $|n\rangle_{g/h}$. The initial state~($m_f=0$ condensate) 
$|\psi(0)\rangle=|0\rangle_g\otimes|0\rangle_h$ evolves as $|\psi(t)\rangle=|\psi(t)\rangle_g\otimes|\psi(t)\rangle_h$ with $|\psi(t)\rangle_{g/h}\equiv\exp(-i\hat H_{g/h}t/\hbar)|0\rangle_{g/h}$.
An explicit form of $|\psi(t)\rangle_{g/h}$ is obtained from the disentanglement theorem for $\mathfrak{su}(1,1)$ \cite{Truax1985}: for $t,\,r\in\mathbb{R}$, $z\in\mathbb{C}$
\begin{align}
\begin{split}
&\exp\!\left(\left(z\hat{L}_+-z^*\hat{L}_-+2ir\hat{L}_0\right)t\right) \\
&= \exp\!\left(p_+\hat{L}_+\right)\exp\!\left(p_0\hat{L}_0\right)\exp\!\left(p_-\hat{L}_-\right)
\end{split}
\end{align}
with $p_+ = z\tau\sinh(t/\tau)/ C$, $p_0 = -2\ln C$, $p_- = -z^*\tau\sinh(t/\tau)/C$, $C = \cosh\left(t/\tau\right) - ir\tau\sinh\left(t/\tau\right)$, and
$\tau = \left(|z|^2-r^2\right)^{-1/2}$.
The result for $|z|^2=r^2$ is obtained by taking the corresponding limit \cite{Truax1985}. This entails, see also~\cite{PezzeRMP},
\begin{equation}
|\psi(t)\rangle_{g,h}=\frac{1}{\sqrt{|C|}}\sum_{n=0}^{\infty}\epsilon_{g,h}^n\sqrt{\binom{2n}{n}}\left(\frac{c}{2}\right)^n|2n\rangle_{g,h}
\end{equation}
with $c = -i\beta\tau\sinh(t/\tau) / C$, $C = \cosh\left( t/\tau\right)+i\alpha\tau\sinh \left( t/\tau \right)$, and 
$\tau = \left(\beta^2-\alpha^2\right)^{-1/2}$.

The evaluation of $\langle \hat{N}_g \rangle$, $\Delta\hat{N}_g$, and the QFI relies on the generating function
\begin{equation}
f(z) = \sum_{n=0}^{\infty}\binom{2n}{n}\left(\frac{|c|}{2}\right)^{2n}\!\operatorname{e}^{nz} = \frac{1}{\sqrt{1-|c|^2\operatorname{e}^z}}
\end{equation}
for $|c|^2\operatorname{e}^z<1$. We first derive $\langle \hat{N}_g\rangle = \beta^2\tau^2\sinh^2(t/\tau)$ and $(\Delta \hat{N}_{g})^2=2\langle\hat{N}_g\rangle (\langle\hat{N}_g\rangle+1 )$. Applying $N_0\simeq N$ and $\hat{a}_0^\dagger,\,\hat{a}_0 \simeq \sqrt{N_0}$ also to $\hat{\mathbf{S}}$ gives
\begin{equation}
\hat{S}_x\simeq \frac{\sqrt{N}}{2}(\hat{g}+\hat{g}^\dagger),\;\;\hat{S}_y\simeq \frac{\sqrt{N}}{2i}(\hat{g}-\hat{g}^\dagger),\;\;\hat{S}_z\simeq \frac{1}{2}(N - \hat{g}^\dagger\hat{g}).
\end{equation}
The corresponding covariance matrix $\boldsymbol{\Gamma}_{\hat{\mathbf{S}}}$ is block diagonal. The eigenvalues of its $(x,y)$-block evaluate to
\begin{equation}
\lambda^{(xy)}_\pm=
\frac{N}{4}\left(1+2\langle \hat{N}_g\rangle\pm \sqrt{2}\Delta \hat{N}_g\right),
\end{equation}
while the variance of $\hat{S}_z$ is $\lambda^{(z)}=(\Delta \hat{N}_g)^2/4$. Thus, for $N\gg\langle \hat N_g\rangle$ the maximal QFI attainable with rotations generated by linear combinations of the $\hat S_i$ is $4\lambda_+^{(xy)}$. The corresponding (normalized) eigenvector
\begin{equation}
\mathbf{u}^{(g)}_\mathrm{opt}=
\pm\frac{1}{\sqrt{2}}\left(\sqrt{1+\sqrt{2}\frac{\alpha^2\langle\hat{N}_g\rangle}{\beta^2\Delta\hat{N}_g}},\,-\sqrt{1-\sqrt{2}\frac{\alpha^2\langle\hat{N}_g\rangle}{\beta^2\Delta\hat{N}_g}},\,0\right)
\end{equation} 
gives the optimal direction of interferometric rotations in the symmetric subspace.

Recall that $\phi_{gh}\!: \hat{g}^\dagger \leftrightarrow i\hat{h}^\dagger$ and $\phi_\pm\!: \hat{a}_{1}^\dagger \leftrightarrow \hat{a}_{-1}^\dagger$ leave the algebraic relations of the creation and annihilation operators, the initial state $|\psi(0)\rangle$, the Hamiltonian, and thus $|\psi(t)\rangle$ invariant. Hence $\phi_{gh}(\hat N_g)=\hat N_h$ and $\phi_\pm(\hat N_+)=\hat N_-$ imply $\langle \hat N_h \rangle = \langle \hat N_g \rangle$, $ \Delta \hat N_h = \Delta \hat N_g$, $\langle \hat N_+ \rangle = \langle \hat N_- \rangle$, and $ \Delta \hat N_+ = \Delta \hat N_-$. Then $\hat{N}_g+\hat{N}_h=\hat{N}_++\hat{N}_-$ yields $\langle \hat N_+ \rangle = \langle \hat N_g \rangle$. Using $\hat D|\psi(t)\rangle=0$ and $|\psi(t)\rangle =|\psi(t)\rangle_g\otimes |\psi(t)\rangle_h$ we furthermore find that $\Delta \hat N_+=\Delta \hat N_g/\sqrt{2}$. Finally, 
$\phi_{gh}(\hat A_x)=-\hat S_y$, $\phi_{gh}(\hat A_y)=\hat S_x$, and $\phi_{gh}(\hat A_z)=\hat S_z$ entail that the eigenvalues of $\boldsymbol{\Gamma}_{\hat{\mathbf{A}}}$ coincide with the ones of $\boldsymbol{\Gamma}_{\hat{\mathbf{S}}}$, while
\begin{equation}
\mathbf{u}^{(h)}_\mathrm{opt}=\left(-u^{(g)}_{\mathrm{opt},\,y},\,u^{(g)}_{\mathrm{opt},\,x},\,0\right)
\end{equation}
defines, again for $N\gg\langle \hat N_\pm\rangle$, the optimal rotation for phase imprinting within the antisymmetric subspace.


\bibliographystyle{prsty}

\begin{thebibliography}{}
\bibitem{nota}
We assume $\lambda = \tfrac{4 \pi \hbar^2 (a_2 - a_0)}{3m}
 \int \!\operatorname{d}\!^3 {r} |\phi({\bf r})|^{4}$, where 
$\phi({\bf r})$ is the condensate (Gross-Pitaevskii)
wave function normalized such that $\int \!\operatorname{d}\!^3r |\phi({\bf r})|^2 = 1$, $m$ is the atomic mass and $a_G$ are the scattering lengths for s-wave collisions in the $G = 0,2$ allowed channels.
\end{thebibliography}

\end{document}